\documentclass[journal,a4paper]{IEEEtran}
\usepackage{amssymb}
\usepackage{amsmath}
\usepackage{amsfonts}
\usepackage{float}
\usepackage{graphics}
\usepackage[tight,footnotesize]{subfigure}
\usepackage{rotating}
\usepackage{algorithm}
\usepackage{algorithmic}
\usepackage{enumerate}
\usepackage{cite}
\usepackage{stfloats}
\usepackage{multicol}
\usepackage{array}

\DeclareMathOperator{\D}{d} \DeclareMathOperator{\E}{E}
\DeclareMathOperator{\Tr}{Tr}

\newcommand*{\QEDA}
{\hfill\ensuremath{\blacksquare}}
\DeclareMathSizes{10}{9}{7}{5}

\newtheorem{theorem}{Theorem}

\newtheorem{corollary}{Corollary}

\newtheorem{remark}{Remark}

\input{tcilatex}
\begin{document}

\title{Guaranteeing Positive Secrecy Capacity with Finite-Rate Feedback
using Artificial Noise}
\pubid{}
\specialpapernotice{}
\author{Shuiyin~Liu,~Yi~Hong,~and~Emanuele~Viterbo\thanks{%
S.~Liu, Y.~Hong and E.~Viterbo are with the Department of Electrical and
Computer Systems Engineering, Monash University, Clayton, VIC 3800,
Australia (e-mail: \{shuiyin.liu, yi.hong, emanuele.viterbo\}@monash.edu).
This work was performed at the Monash Software Defined Telecommunications
Lab and the authors were supported by the Australian Research Council
Discovery Project with ARC DP130100336.}}
\maketitle

\begin{abstract}
While the impact of finite-rate feedback on the capacity of fading
channels has been extensively studied in the literature, not much
attention has been paid to this problem under secrecy constraint. In
this work, we study the ergodic secret capacity of a multiple-input
multiple-output multiple-antenna-eavesdropper (MIMOME) wiretap
channel with quantized channel state information (CSI) at the
transmitter and perfect CSI at the legitimate receiver, under the
assumption that only the statistics of eavesdropper CSI is known at
the transmitter. We refine the analysis of the \emph{random vector
quantization} (RVQ) based artificial noise (AN) scheme in
\cite{Shihc11}, where a \emph{heuristic} upper bound on the secrecy
rate loss, when compared to the perfect CSI case, was given. We
propose a lower bound on the ergodic secrecy capacity. We show that
the lower bound and the secrecy capacity with perfect CSI coincide
asymptotically as the number of feedback bits and the AN power go to
infinity. For practical applications, we propose a very efficient
quantization codebook construction method for the two transmit
antennas case.
\end{abstract}

\begin{IEEEkeywords}
artificial noise, secret capacity, physical layer security, wiretap
channel.
\end{IEEEkeywords}

\section{Introduction}

Complexity-based cryptographic technologies (e.g. AES \cite{Daemen:2002:DRA}%
) have traditionally been used to provide a secure gateway for
communications and data exchanges at the network layer. The security is
achieved if an eavesdropper (Eve) without the key cannot decipher the
message in a reasonable amount of time. This premise becomes controversial
with the rapid developments of computing devices (e.g. quantum computer). In
contrast, physical layer security (PLS) does not depend on a specific
computational model and can provide security even when Eve has unlimited
computing power. Wyner \cite{Wyner75} and later Csisz\'{a}r and K\"{o}rner
\cite{Csiszar78} proposed the \emph{wiretap channel} model as a basic
framework for PLS. Wyner has shown that for discrete memoryless channels, if
Eve intercepts a degraded version of the intended receiver's (Bob's) signal,
a prescribed degree of data confidentiality could be simultaneously attained
by channel coding without any secret key. The associated notion of \emph{%
secrecy capacity} was introduced to characterize the maximum transmission
rate from the transmitter (Alice) to Bob, below which Eve is unable to
obtain any information.

Wyner's wiretap channel model has been extended to fading channel \cite%
{Gopala08}, Gaussian broadcast channel \cite{RuohengLiu09}, multiple-input
single-output multiple-antenna-eavesdropper (MISOME) channel \cite%
{Khisti10SO}, and multiple-input multiple-output
multiple-antenna-eavesdropper (MIMOME) channel \cite{Oggier11}. All these
works rely on the perfect knowledge of Bob's channel state information (CSI)
at Alice to compute the secrecy capacity and enable secure encoding. In
particular, Eve's CSI is also assumed to be known at Alice in \cite%
{Gopala08, RuohengLiu09, Oggier11}, although the CSI of a passive Eve is
very hard to be unveiled at Alice. It is more reasonable to assume that
Alice only knows the statistics of Eve's channel. Even the assumption of
perfect knowing Bob's CSI is not realistic. In practice, Bob can only
provide Alice with a quantized version of his CSI via a rate constrained
feedback channel (i.e., finite-rate feedback).

In this work, we are interested in the secrecy capacity conditioned on the
quantized CSI of Bob's channel and the statistics of Eve's channel. While
the impact of finite-rate feedback on the capacity of fading channels has
been extensively studied (see \cite{Love03, Mukkavilli03, Dai08,
Santipach09,Pitaval11}), not much attention has been given to this problem
under secrecy constraint. In \cite{Rezki14}, assuming that Alice only knows
the statistics of Eve's channel, the authors derived lower and upper bounds
on the ergodic secrecy capacity for a single-input single-output
single-antenna-eavesdropper (SISOSE) system with finite-rate feedback of
Bob's CSI. In the MIMOME scenario, the artificial noise (AN) scheme has been
shown to guarantee positive secrecy capacity without knowing Eve's CSI in
\cite{Goel08}. Alice is assumed to have perfect knowledge of Bob's
eigenchannel vectors. This assumption allows her to align artificial noise
within the null space of a MIMO channel between Alice and Bob, so that only
Eve's equivocation is enhanced. In \cite{Shihc11}, the authors show that if
only quantized CSI is available at Alice, the artificial noise will leak
into Bob's channel, causing a decrease in the achievable secrecy rate. A
\emph{heuristic} upper bound on the secrecy rate loss (compared to the
perfect CSI case) is proposed in \cite[Eq. 34]{Shihc11}.

The main contribution of this paper is to provide a lower bound on the
ergodic secrecy capacity for the AN scheme with quantized CSI, valid for any
number of Alice/Bob/Eve antennas, as well as for any Bob/Eve SNR regimes.
Following the work in \cite{Shihc11}, we use the \emph{random vector
quantization} (RVQ) scheme in \cite{Love03}. Namely, given $B$ feedback
bits, Bob quantizes his eigenchannel matrix to one of $N=2^{B}$ random
unitary matrices and feeds back the corresponding index. We first show that
RVQ is asymptotically optimal for security purpose, i.e., the secrecy
capacity/rate loss compared to the perfect CSI case converges to $0$ as $%
B\rightarrow \infty $. This result implies that the heuristic bound in \cite[%
Eq. 34]{Shihc11} is not tight, since it reduces to a positive constant as $%
B\rightarrow \infty $. To refine the analysis in \cite{Shihc11}, we
establish a tighter upper bound on the secrecy rate loss, which leads to an
explicit lower bound on the ergodic secrecy capacity. We further show that
the lower bound and the secrecy capacity with perfect CSI coincide
asymptotically as $B$ and the AN power go to infinity. This allows us to
provide a sufficient condition guaranteeing positive secrecy capacity.

From a practical point of view, it is often desirable to use a deterministic
quantization codebook rather than a random one. The problem of derandomizing
RVQ codebooks is related to discretizing the complex Grassmannian manifold
\cite{Love03,Mukkavilli03}. Since the optimal constructions are possible
only in very special cases, deterministic codebooks are mostly generated by
computer search \cite{Xia06}. Interestingly, the case of codebook design
with two transmit antennas is equivalent to quantizing a real sphere \cite%
{Pitaval11}. According to this fact, we propose a very efficient codebook
construction method for the two-antenna case. Simulation results demonstrate
that near-RVQ performance is achieved by a moderate number of feedback bits.

The paper is organized as follows: Section II presents the system model,
followed by the analysis of secrecy capacity with finite-rate feedback in
Section III. Section IV provides the deterministic quantization codebook
construction method for the two-antenna case. Conclusions are drawn in
Section V. Proofs of the theorems are given in Appendix.

\textit{Notation:} Matrices and column vectors are denoted by upper and
lowercase boldface letters, and the Hermitian transpose, inverse,
pseudoinverse of a matrix $\mathbf{B}$ by $\mathbf{B}^{H}$, $\mathbf{B}^{-1}$%
, and $\mathbf{B}^{\dagger }$, respectively. $|\mathbf{B}|$ denotes the
determinant of $\mathbf{B}$. Let the random variables $\left\{ X_{n}\right\}
$ and $X$ be defined on the same probability space. We write $X_{n}\overset{%
a.s.}{\rightarrow }X$ if $X_{n}$ converges to $X$ almost surely or with
probability one. $\mathbf{I}_{n}$ denotes the identity matrix of size $n$.
An $m\times n$ null matrix is denoted by $\mathbf{0}_{m\times n}$. A
circularly symmetric complex Gaussian random variable $x$ with variance $%
\sigma ^{2}$\ is defined as $x\backsim \mathcal{N}_{\mathbb{C}}(0,\sigma
^{2})$. The real, complex, integer and complex integer numbers are denoted
by $\mathbb{R}$, $\mathbb{C}$, $\mathbb{Z}$ and $\mathbb{Z}\left[ i\right] $%
, respectively. $I(x;y)$ represents the mutual information of two random
variables $x$ and $y$. We use the standard asymptotic notation $f\left(
x\right) =O\left( g\left( x\right) \right) $ when $\lim
\sup\limits_{x\rightarrow \infty }|f(x)/g(x)|<\infty $. $\lceil x\rfloor $
rounds to the closest integer, while $\lfloor x\rfloor $ to the closest
integer smaller than or equal to $x$ and $\left\lceil x\right\rceil $ to the
closest integer larger than or equal to $x $. A central complex Wishart
matrix $\mathbf{A}\in \mathbb{C}^{m\times m}$ with $n$ degrees of freedom
and covariance matrix $\mathbf{\Sigma }$, is defined as $\mathbf{A}\backsim
W_{m}(n$,$\mathbf{\Sigma })$. Trace of a square matrix $\mathbf{B}$ is
denoted by $\Tr\left( \mathbf{B}\right) $. We write $\triangleq $ for
equality in definition.

\section{System Model}

We consider secure communications over a three-terminal system, including\ a
transmitter (Alice), the intended receiver (Bob), and an unauthorized
receiver (Eve), equipped with $N_{\text{A}}$, $N_{\text{B}}$, and $N_{\text{E%
}}$ antennas, respectively. The signal vectors received by Bob and Eve are%
\begin{align}
\mathbf{z}& =\mathbf{Hx}+\mathbf{n}_{\text{B}}\text{,}  \label{B1} \\
\mathbf{y}& =\mathbf{G\mathbf{\mathbf{x}}}+\mathbf{n}_{\text{E}}\text{,}
\label{E1}
\end{align}%
where $\mathbf{x}\in \mathbb{C}^{N_{\text{A}}}$ is the transmit signal
vector, $\mathbf{H}\in \mathbb{C}^{N_{\text{B}}\times N_{\text{A}}}$ and $%
\mathbf{G}\in \mathbb{C}^{N_{\text{E}}\times N_{\text{A}}}$ are the
respective channel matrices between Alice to Bob and Alice to Eve, and $%
\mathbf{n}_{\text{B}}$, $\mathbf{n}_{\text{E}}$ are AWGN vectors with i.i.d.
entries $\sim \mathcal{N}_{\mathbb{C}}(0$, $\sigma _{\text{B}}^{2})$ and $%
\mathcal{N}_{\mathbb{C}}(0$, $\sigma _{\text{E}}^{2})$. We assume that the
entries of $\mathbf{H}$ and $\mathbf{G}$ are i.i.d. complex random variables
$\sim \mathcal{N}_{\mathbb{C}}(0$, $1)$.

Without loss of generality, we normalize Bob's channel noise variance to
one, i.e.,%
\begin{equation}
\sigma _{\text{B}}^{2}=1\text{.}
\end{equation}

In this paper, we assume that Bob knows its own channel matrix $\mathbf{H}$
instantaneously and Eve knows both its own channel matrix $\mathbf{G}$ and
the main channel $\mathbf{H}$, instantaneously; whereas Alice is only aware
of the statistics of $\mathbf{H}$ and $\mathbf{G}$. There is also an
error-free public feedback channel with limited capacity from Bob to Alice
that can be tracked by Eve. In our setting, the feedback is exclusively used
to send the index of the codeword in a quantization codebook that describes
the main channel state information $\mathbf{H}$. The quantization codebook
is assumed to be known \emph{a priori} to Alice, Bob and Eve.

\subsection{Artificial Noise Scheme with Perfect CSI}

The original AN scheme assumes $N_{\text{B}}<N_{\text{A}}$, in order to
ensure that $\mathbf{H}$ has a non-trivial null space with an orthonormal
basis $\mathbf{Z}=\mbox{null}(\mathbf{H})\in \mathbb{C}^{N_{\text{A}}\times
(N_{\text{A}}-N_{\text{B}})}$ (such that $\mathbf{HZ}=\mathbf{0}_{N_{\text{B}%
}\times (N_{\text{A}}-N_{\text{B}})}$) \cite{Goel08}. Let $\mathbf{H}=%
\mathbf{U}\mathbf{\Lambda }\mathbf{V}^{H}$ be the singular value
decomposition (SVD) of $\mathbf{H}$, where $\mathbf{U}\in \mathbb{C}^{N_{%
\text{B}}\times N_{\text{B}}}$ and $\mathbf{V}\in \mathbb{C}^{N_{\text{A}%
}\times N_{\text{A}}}$ are unitary matrices. Then, we can write the unitary
matrix $\mathbf{V}$ as%
\begin{equation}
\mathbf{V=[\tilde{V}},\mathbf{Z]}\text{,}  \label{T0}
\end{equation}%
where the $N_{\text{B}}$ columns of $\mathbf{\tilde{V}}\in \mathbb{C}^{N_{%
\text{A}}\times N_{\text{B}}}$ span the orthogonal complement subspace to
the null space spanned by the columns of $\mathbf{Z}\in \mathbb{C}^{N_{\text{%
A}}\times (N_{\text{A}}-N_{\text{B}})}$.

With unlimited feedback (i.e., perfect CSI), Alice has perfect knowledge of
the precoding matrix $\mathbf{V}$, and transmits%
\begin{equation}
\mathbf{x}=\mathbf{\tilde{V}}\mathbf{u}+\mathbf{Zv}=\mathbf{V}\left[
\begin{array}{c}
\mathbf{u} \\
\mathbf{v}%
\end{array}%
\right] \text{,}  \label{T1}
\end{equation}%
where $\mathbf{u}\in \mathbb{C}^{N_{\text{B}}}$ is the information vector
and $\mathbf{v}\in \mathbb{C}^{(N_{\text{A}}-N_{\text{B}})}$ is the
\textquotedblleft artificial noise\textquotedblright . For the purpose of
evaluating the achievable secrecy rate, both $\mathbf{u}$ and $\mathbf{v}$
are assumed to be circular symmetric Gaussian random vectors with i.i.d.
complex entries $\sim \mathcal{N}_{\mathbb{C}}(0$, $\sigma _{\text{u}}^{2})$
and $\mathcal{N}_{\mathbb{C}}(0$, $\sigma _{\text{v}}^{2})$, respectively.
In \cite{shuiyin14_2}, we have shown that Gaussian input alphabets
asymptotically achieves the secrecy capacity as $\sigma _{\text{v}%
}^{2}\rightarrow \infty $.

Equations (\ref{B1}) and (\ref{E1}) can then be rewritten as%
\begin{align}
\mathbf{z}& =\mathbf{H\mathbf{\mathbf{\tilde{V}}}}\mathbf{u}+\mathbf{HZv+n}_{%
\text{B}}=\mathbf{H\mathbf{\mathbf{\tilde{V}}}}\mathbf{u}+\mathbf{n}_{\text{B%
}}\text{,}  \label{sec_mod2} \\
\mathbf{y}& =\mathbf{G\mathbf{\mathbf{\tilde{V}}}}\mathbf{u}+\mathbf{GZ}%
\mathbf{v}+\mathbf{n}_{\text{E}}\text{,}  \label{Eve_mod2}
\end{align}%
to show that with unlimited feedback, the artificial noise only degrades
Eve's channel, resulting in increased secrecy capacity (compared to the
non-AN case).

\subsection{Artificial Noise Scheme with Quantized CSI}

In \cite{Shihc11}, the authors analyzed the impact of finite-rate feedback
on the secrecy rate achievable by the AN scheme. To quantize the matrix $%
\mathbf{\tilde{V}}$ in (\ref{T0}), the random vector quantization (RVQ)
scheme in \cite{Love03} is used. Given $B$ feedback bits per fading channel,
Bob specifies $\mathbf{\tilde{V}}$ from a random quantization codebook%
\begin{equation}
\mathcal{V}=\left\{ \mathbf{\tilde{V}}_{i},1\leq i\leq 2^{B}\right\} \text{,}
\label{V}
\end{equation}%
where the entries are independent $N_{\text{A}}\times N_{\text{B}}$ random
unitary matrices, i.e., $\mathbf{\tilde{V}}_{i}^{H}\mathbf{\tilde{V}}_{i}=%
\mathbf{I}_{N_{\text{B}}}$. The codebook $\mathcal{V}$ is known \emph{a
priori} to both Alice, Bob and Eve. Bob selects the $\mathbf{\tilde{V}}_{j}$
that minimize the chordal distance between $\mathbf{\tilde{V}}_{i}$ and $%
\mathbf{\tilde{V}}$ \cite{Dai08}:%
\begin{equation}
\mathbf{\tilde{V}}_{j}=\min_{\mathbf{\tilde{V}}_{i}\in \mathcal{V}%
}d^{2}\left( \mathbf{\tilde{V}}_{i}\text{,}\mathbf{\tilde{V}}\right) \text{,}
\label{cost_fun}
\end{equation}%
where
\begin{equation}
d\left( \mathbf{\tilde{V}}_{i}\text{,}\mathbf{\tilde{V}}\right) =N_{\text{B}%
}-\Tr\left( \mathbf{\tilde{V}}^{H}\mathbf{\tilde{V}}_{i}\mathbf{\tilde{V}}%
_{i}^{H}\mathbf{\tilde{V}}\right) \text{.}  \label{Dist}
\end{equation}%
Note that $\Tr\left( \mathbf{A}\right) $ denotes the trace of the square
matrix $\mathbf{A}$. And then, Bob relays the corresponding index $j$ back
to Alice.

Alice generates the precoding matrix from $\mathbf{\tilde{V}}_{j}$ as
follows. Let $\mathbf{\tilde{v}}_{1}$, ... , $\mathbf{\tilde{v}}_{N_{\text{B}%
}}$ be the columns of $\mathbf{\tilde{V}}_{j}$, and $\mathbf{e}_{1}$, ... , $%
\mathbf{e}_{N_{\text{A}}-N_{\text{B}}}$ be the standard basis vectors. Alice
applies the Gram-Schmidt algorithm to the matrix%
\begin{equation*}
\left[ \mathbf{\tilde{v}}_{1}\text{, ... , }\mathbf{\tilde{v}}_{N_{\text{B}}}%
\text{, }\mathbf{e}_{1}\text{, ... , }\mathbf{e}_{N_{\text{A}}-N_{\text{B}}}%
\right]
\end{equation*}%
to generate the remaining orthonormal basis vectors spanning the orthogonal
complement space to the one generated by the columns of $\mathbf{\tilde{V}}%
_{j}$. This provides Alice with a unitary matrix%
\begin{equation}
\mathbf{\hat{V}=[\tilde{V}}_{j},\mathbf{\hat{Z}]}\in \mathbb{C}^{N_{\text{A}%
}\times N_{\text{A}}}\text{,}  \label{V_hat}
\end{equation}%
that can be used to precode $\mathbf{u}$ and $\mathbf{v}$ as in (\ref{T1}).

Since $\mathbf{\hat{Z}\neq Z}$, the interference term $\mathbf{H\hat{Z}v}$
cannot be nulled at Bob. Therefore, equations (\ref{sec_mod2}) and (\ref%
{Eve_mod2}) reduce to%
\begin{align}
\mathbf{z}& =\mathbf{H\tilde{V}}_{j}\mathbf{u}+\mathbf{H\hat{Z}v+n}_{\text{B}%
}\text{,}  \label{sec_mod3} \\
\mathbf{y}& =\mathbf{G\tilde{V}}_{j}\mathbf{u}+\mathbf{G\hat{Z}v}+\mathbf{n}%
_{\text{E}}\text{,}  \label{Eve_mod3}
\end{align}%
and show that with finite rate feedback (i.e., quantized CSI), some of the
artificial noise will inevitably leak into the main channel from Alice to
Bob, causing degradation in the secrecy capacity (compared to the unlimited
feedback case).

\subsection{Assumptions and Notations}

The analysis in \cite{Goel08, Shihc11} are based on the assumption of $N_{%
\text{E}}<N_{\text{A}}$. Clearly, this assumption is not always realistic.
In this work, we remove this assumption and evaluate the secrecy capacity
for any number of Eve antennas.

Since $\mathbf{\hat{V}}$ in (\ref{V_hat}) is a unitary matrix, the total
transmission power can be written as%
\begin{equation}
||\mathbf{x}||^{2}=\left[
\begin{array}{c}
\mathbf{u} \\
\mathbf{v}%
\end{array}%
\right] ^{H}\mathbf{\hat{V}}^{H}\mathbf{\hat{V}}\left[
\begin{array}{c}
\mathbf{u} \\
\mathbf{v}%
\end{array}%
\right] =||\mathbf{u}||^{2}+||\mathbf{v}||^{2}\text{.}  \label{Total_power}
\end{equation}%
Then the average transmit power constraint $P$ is%
\begin{equation}
P=\text{E}(||\mathbf{x}||^{2})=P_{\text{u}}+P_{\text{v}}\text{,}
\label{Power_Cons}
\end{equation}%
where%
\begin{equation}
\begin{array}{lll}
P_{\text{u}} & = & \text{E}(||\mathbf{u}||^{2})=\sigma _{\text{u}}^{2}N_{%
\text{B}}\text{,} \\
P_{\text{v}} & = & \text{E}(||\mathbf{v}||^{2})=\sigma _{\text{v}}^{2}(N_{%
\text{A}}-N_{\text{B}})\text{,}%
\end{array}
\label{Power_c2}
\end{equation}%
are fixed by the power allocation scheme that selects the power balance
between $\sigma _{\text{u}}^{2}$ and $\sigma _{\text{v}}^{2}$.

We define Bob's and Eve's SNRs as

\begin{itemize}
\item SNR$_{\text{B}}\triangleq \sigma _{\text{u}}^{2}/\sigma _{\text{B}%
}^{2} $

\item SNR$_{\text{E}}\triangleq \sigma _{\text{u}}^{2}/\sigma _{\text{E}%
}^{2} $
\end{itemize}

To simplify our notation, we define three system parameters:

\begin{itemize}
\item $\alpha \triangleq \sigma _{\text{u}}^{2}/\sigma _{\text{E}}^{2}=$ SNR$%
_{\text{E}}$

\item $\beta \triangleq \sigma _{\text{v}}^{2}/\sigma _{\text{u}}^{2}$ (AN
power allocation)

\item $\gamma \triangleq \sigma _{\text{E}}^{2}/\sigma _{\text{B}}^{2}$
(Eve-to-Bob noise-power ratio)
\end{itemize}

Note that SNR$_{\text{B}}=\alpha \gamma $. If $\gamma >1$, we say Eve has a
\emph{degraded} channel. Since we have normalized $\sigma _{\text{B}}^{2}$
to one, we\ can rewrite (\ref{Power_c2}) as

\begin{itemize}
\item $P_{\text{u}}=\alpha \gamma N_{\text{B}}$

\item $P_{\text{v}}=\alpha \beta \gamma (N_{\text{A}}-N_{\text{B}})$
\end{itemize}

\subsection{Instantaneous and Ergodic Secrecy Capacities}

We recall from \cite{Oggier11} the definition of instantaneous secrecy
capacity for MIMOME channel:%
\begin{equation}
C_{\text{S}}\triangleq \max_{p\left( \mathbf{u}\right) }\left\{ I(\mathbf{%
u;z)-}I(\mathbf{u;y)}\right\} \text{.}  \label{CS}
\end{equation}%
where the maximum is taken over all possible input distributions $p\left(
\mathbf{u}\right) $. We remark that $C_{\text{S}}$ is a function of $\mathbf{%
H}$ and $\mathbf{G}$, which are embedded in $\mathbf{z}$ and $\mathbf{y}$.
To average out the channel randomness, we further define the ergodic secrecy
capacity, as in \cite{Goel08}%
\begin{equation}
\E(C_{\text{S}})\triangleq \max_{p\left( \mathbf{u}\right) }\left\{ I(%
\mathbf{u;z|H)-}I(\mathbf{u;y|H}\text{, }\mathbf{G)}\right\} \text{,}
\label{CS_ave}
\end{equation}%
where $I\left( X;Y|Z\right) \triangleq \text{E}_{Z}\left[ I\left( X;Y\right)
|Z\right] $, following the notation in \cite{Telatar99}.

Since closed form expressions for $C_{\text{S}}$ and $\E(C_{\text{S}})$ are
not always available, we often consider the corresponding secrecy rates,
given by%
\begin{equation}
R_{\text{S}}\triangleq I(\mathbf{u;z)-}I\mathbf{(\mathbf{u;y)}}\text{,}
\label{CLB}
\end{equation}%
\begin{equation}
\E(R_{\text{S}})\triangleq I(\mathbf{u;z|H)-}I\mathbf{(\mathbf{u;y|H},%
\mathbf{G)}}\text{,}  \label{CLB_ave}
\end{equation}%
assuming Gaussian input alphabets, i.e., $\mathbf{v}$ and $\mathbf{u}$ are
mutually independent Gaussian vectors with i.i.d. complex entries $\mathcal{N%
}_{\mathbb{C}}(0$, $\sigma _{\text{v}}^{2})$ and $\mathcal{N}_{\mathbb{C}}(0$%
, $\sigma _{\text{u}}^{2})$, respectively.

From (\ref{sec_mod2}), (\ref{Eve_mod2}) and (\ref{CLB}), the achievable
secrecy rate with perfect CSI can be written as%
\begin{align}
R_{\text{S}}& =\log \left\vert \mathbf{I}_{N_{\text{B}}}+\alpha \gamma
\mathbf{HH}^{H}\right\vert +\log \left\vert \mathbf{I}_{N_{\text{E}}}+\alpha
\beta (\mathbf{GZ)}(\mathbf{GZ)}^{H}\right\vert  \notag \\
& -\log \left\vert \mathbf{I}_{N_{\text{E}}}+\alpha (\mathbf{G\tilde{V}})(%
\mathbf{G\tilde{V}})^{H}+\alpha \beta (\mathbf{GZ)}(\mathbf{GZ)}%
^{H}\right\vert \text{.}  \label{SR_P_CSI}
\end{align}%
The closed-form expression of $\E(R_{\text{S}})$ can be found in \cite[Th. 1]%
{shuiyin14_2}. It is shown that $\E(R_{\text{S}})\rightarrow \E(C_{\text{S}%
}) $ as the AN power $P_{\text{v}}$ $\rightarrow \infty $ in \cite[Th. 3]%
{shuiyin14_2}.

From (\ref{sec_mod3}), (\ref{Eve_mod3}) and (\ref{CLB}), the achievable
secrecy rate with quantized CSI can be written as%
\begin{align}
R_{\text{S,Q}}& =\log \dfrac{\left\vert \mathbf{I}_{N_{\text{B}}}+\alpha
\gamma (\mathbf{H\tilde{V}}_{j})(\mathbf{H\tilde{V}}_{j})^{H}+\alpha \beta
\gamma (\mathbf{H\hat{Z})}(\mathbf{H\hat{Z})}^{H}\right\vert }{\left\vert
\mathbf{I}_{N_{\text{B}}}+\alpha \beta \gamma (\mathbf{H\hat{Z})}(\mathbf{H%
\hat{Z})}^{H}\right\vert }  \notag \\
& -\log \dfrac{\left\vert \mathbf{I}_{N_{\text{E}}}+\alpha (\mathbf{GV}%
_{1,j})(\mathbf{G\tilde{V}}_{j})^{H}+\alpha \beta (\mathbf{G\hat{Z})}(%
\mathbf{G\hat{Z})}^{H}\right\vert }{\left\vert \mathbf{I}_{N_{\text{E}%
}}+\alpha \beta (\mathbf{G\hat{Z})}(\mathbf{G\hat{Z})}^{H}\right\vert }\text{%
.}  \label{SR_Q_CSI}
\end{align}

\subsection{Open Problems and Motivations}

Using \cite[Eq. 2, pp. 56]{Helmut96}, it is simple to show that%
\begin{equation}
\E(R_{\text{S}})\geq \E(R_{\text{S,Q}})\text{.}
\end{equation}%
In \cite{Shihc11}, the \emph{ergodic secrecy rate loss} is defined by:%
\begin{equation}
\E(\Delta R_{\text{S}})\triangleq \E(R_{\text{S}})-\E(R_{\text{S,Q}})\text{.}
\label{D_R}
\end{equation}%
A heuristic upper bound was proposed in \cite[Eq. 34]{Shihc11}:%
\begin{align}
\E(\Delta R_{\text{S}})& \lessapprox N_{\text{B}}\log \left( \dfrac{N_{\text{%
B}}+\alpha \beta \gamma N_{\text{A}}D\left( N_{\text{A}},N_{\text{B}%
},2^{B}\right) }{N_{\text{B}}-D\left( N_{\text{A}},N_{\text{B}},2^{B}\right)
}\right)  \notag \\
& +N_{\text{B}}\log \left( 1+\dfrac{1}{\alpha \gamma \left( N_{\text{A}}-N_{%
\text{B}}\right) }\right) \triangleq \text{UB}_{\text{heuristic}}\text{.}
\label{Old_UB}
\end{align}%
where%
\begin{equation}
D\left( N_{\text{A}},N_{\text{B}},2^{B}\right) =\E\left( d^{2}\left( \mathbf{%
\tilde{V}}_{j}\text{,}\mathbf{\tilde{V}}\right) \right) \text{,}  \label{DD}
\end{equation}%
and $d\left( \cdot ,\cdot \right) $ is give in (\ref{Dist}).

However, (\ref{Old_UB}) is insufficient to characterize the impact of
quantized CSI on the secrecy rate achievable by the AN scheme. To see this,
we first provide the following theorem.

\begin{theorem}
\label{Th1}For the RVQ-based AN scheme, as $B\rightarrow \infty $,%
\begin{equation}
\mathbf{\hat{V}\rightarrow V}\text{,}
\end{equation}%
where $\mathbf{V}$ is given in (\ref{T0}) and $\mathbf{\hat{V}}$ is given in
(\ref{V_hat}).
\end{theorem}

\begin{proof}
See Appendix A.
\end{proof}

Theorem \ref{Th1} shows the RVQ scheme is asymptotically optimal for large $%
B $, i.e., the secrecy capacity/rate loss (compared to the perfect CSI case)
converges to zero. In contrast, as $B\rightarrow \infty $, UB$_{\text{%
heuristic}}$ in (\ref{Old_UB}) reduces to a positive constant:%
\begin{equation}
\text{UB}_{\text{heuristic}}\rightarrow N_{\text{B}}\log \left( 1+\dfrac{1}{%
\alpha \gamma \left( N_{\text{A}}-N_{\text{B}}\right) }\right) \text{,}
\label{uc}
\end{equation}%
since $D\left( N_{\text{A}},N_{\text{B}},2^{B}\right) \rightarrow 0$ as $%
B\rightarrow \infty $ \cite{Dai08}. Hence, the heuristic bound in (\ref%
{Old_UB}) is not tight.

\begin{remark}
The ergodic secrecy capacity with quantized CSI, denoted by $\E(C_{\text{S,Q}%
})$, is lower bounded by%
\begin{equation}
\E(C_{\text{S,Q}})\geq \E(R_{\text{S,Q}})=\E(R_{\text{S}})-\E(\Delta R_{%
\text{S}})\text{.}  \label{Bound_M}
\end{equation}%
Using the closed-form expression of $\E(R_{\text{S}})$ given in \cite[Th. 1]%
{shuiyin14_2}, we are motivated to establish a tighter upper bound on $\E%
(\Delta R_{\text{S}})$, which allows us to obtain a lower bound on $\E(C_{%
\text{S,Q}})$.
\end{remark}

\section{Secrecy Capacity with Quantized CSI}

In this section, we wish to determine the secrecy capacity with RVQ scheme.
A tight upper bound on the ergodic secrecy rate loss and a lower bound on
the ergodic secrecy capacity are provided in Theorem 2 and Theorem 4,
respectively. In Theorem 5, we show that the lower bound and the secrecy
capacity with perfect/quantized CSI coincide asymptotically as $B$ and $P_{%
\text{v}}$ go to infinity. This provides a sufficient condition guaranteeing
positive secrecy capacity.

To describe our result, we first recall the following function from \cite%
{shin03}:%
\begin{align}
& \Theta (m,n,x)\triangleq
e^{-1/x}\sum_{k=0}^{m-1}\sum_{l=0}^{k}\sum_{i=0}^{2l}\Bigg\{\dfrac{%
\displaystyle(-1)^{i}(2l)!(n-m+i)!}{\displaystyle2^{2k-i}l!i!(n-m+l)!}
\notag \\
& \cdot \left( \!\!\!%
\begin{array}{c}
2(k-l) \\
k-l%
\end{array}%
\!\!\!\right) \cdot \left( \!\!\!%
\begin{array}{c}
2(l+n-m) \\
2l-i%
\end{array}%
\!\!\!\right) \cdot \sum_{j=0}^{n-m+i}x^{-j}\Gamma (-j,1/x)\Bigg\}\text{,}
\label{Cap_CF}
\end{align}%
where
\scalebox{1}[0.8]{$\Big(\!\!\!
\begin{array}{c}
a \\
b\end{array}\!\!\!\Big)$} $=a!/((a-b)!b!)$ is the binomial coefficient, $%
n\geq m$ are positive integers, and $\Gamma (a,b)$ is the incomplete Gamma
function%
\begin{equation}
\Gamma (a,b)=\int_{b}^{\infty }x^{a-1}e^{-x}dx\text{.}
\end{equation}

We further define%
\begin{eqnarray}
N_{\min } &\triangleq &\min \left\{ N_{\text{E}}\text{, }N_{\text{A}}-N_{%
\text{B}}\right\} \text{,}  \label{n_min} \\
N_{\max } &\triangleq &\max \left\{ N_{\text{E}}\text{, }N_{\text{A}}-N_{%
\text{B}}\right\} \text{,}  \label{n_max}
\end{eqnarray}%
\begin{eqnarray}
\hat{N}_{\min } &\triangleq &\min \left\{ N_{\text{E}}\text{, }N_{\text{A}%
}\right\} \text{,}  \label{n_hat_min} \\
\hat{N}_{\max } &\triangleq &\max \left\{ N_{\text{E}}\text{, }N_{\text{A}%
}\right\} \text{.}  \label{n_hat_max}
\end{eqnarray}

Finally, we define a set of $N_{\text{A}}$ power ratios $\left\{ \theta
_{i}\right\} _{1}^{N_{\text{A}}}$, where%
\begin{equation}
\theta _{i}\triangleq \QATOPD\{ . {\alpha ~~~~~~~~~~1\leq i\leq N_{\text{B}%
}}{\alpha \beta ~~~N_{\text{B}}+1\leq i\leq N_{\text{A}}}  \label{theta}
\end{equation}

We recall from \cite[Th. 4]{Dai08} that%
\begin{equation}
\mu \left( N_{\text{A}},N_{\text{B}},2^{B}\right) \leq D\left( N_{\text{A}%
},N_{\text{B}},2^{B}\right) \leq \eta \left( N_{\text{A}},N_{\text{B}%
},2^{B}\right) \text{,}  \label{bbound}
\end{equation}%
where $D\left( \cdot ,\cdot ,\cdot \right) $ is given in (\ref{DD}) and%
\begin{eqnarray}
\eta \left( n,p,K\right) &=&\dfrac{\Gamma \left( \dfrac{1}{p(n-p)}\right) }{%
p(n-p)}\left( Kc\left( n,p\right) \right) ^{-\dfrac{1}{p(n-p)}}  \notag \\
&&+p\exp (-\left( Kc\left( n,p\right) \right) ^{1-\zeta })\text{,}
\label{niu}
\end{eqnarray}%
\begin{equation}
\mu \left( n,p,K\right) =\dfrac{p(n-p)}{p(n-p)+1}\left( Kc\left( n,p\right)
\right) ^{-\dfrac{1}{p(n-p)}\text{,}}  \label{piu}
\end{equation}%
\begin{equation}
c\left( n,p\right) =\left\{
\begin{array}{l}
\dfrac{1}{\Gamma (p(n-p)+1)}\prod\limits_{i=1}^{p}\dfrac{\Gamma (n-i+1)}{%
\Gamma (p-i+1)}\text{, }~~~~~~~n\geq 2p \\
\dfrac{1}{\Gamma (p(n-p)+1)}\prod\limits_{i=1}^{n-p}\dfrac{\Gamma (n-i+1)}{%
\Gamma (n-p-i+1)}\text{, }~~n\leq 2p%
\end{array}%
\right.  \label{cp}
\end{equation}%
for any $0<\zeta <1$. Note that $\Gamma (a)$ is the Gamma function.

\subsection{Bounds on Ergodic/Instantaneous Secrecy Rate Loss}

We first consider the ergodic secrecy rate loss $\E(\Delta R_{\text{S}})$.

\begin{theorem}
\label{Th2}Let $\theta _{\min }=\min \left\{ \alpha \gamma ,\alpha \beta
\gamma \right\} $. We have%
\begin{align}
& \E(\Delta R_{\text{S}})\leq \Theta (N_{\text{B}},N_{\text{A}},\alpha
\gamma )-\Theta (N_{\text{B}},N_{\text{A}},\theta _{\min })  \notag \\
& +\Theta \left( N_{\text{B}},N_{\text{A}},\alpha \beta \gamma \dfrac{\eta
\left( N_{\text{A}},N_{\text{B}},2^{B}\right) }{N_{\text{B}}}\right)
\triangleq \text{UB,}  \label{MMO}
\end{align}%
where $\Theta (\cdot ,\cdot ,\cdot )$ is given in (\ref{Cap_CF}) and $\eta
(\cdot ,\cdot ,\cdot )$ is given in (\ref{niu}).
\end{theorem}

\begin{proof}
See Appendix B.
\end{proof}

Theorem \ref{Th2} gives a tight upper bound on $\E(\Delta R_{\text{S}})$,
for any number of Alice/Bob/Eve antennas, as well as for any Bob/Eve SNR
regimes. Different from (\ref{uc}), if $\beta \geq 1$, as $B\rightarrow
\infty $,%
\begin{equation}
\text{UB}\rightarrow 0\text{,}
\end{equation}%
which is consistent with Theorem \ref{Th1}.

\begin{example}
Let us apply Theorem \ref{Th2} to the analysis of a RVQ-based AN scheme with
$\beta =1$, $\alpha \gamma =1$, $N_{\text{A}}=4$ and $N_{\text{B}}=2$. The
numerical result in Fig.~1 shows that the proposed upper bound in (\ref%
{MMO}) is much tighter than the heuristic one in (\ref{Old_UB}), and
captures the behavior of $\mathrm{E}(\Delta R_{\text{S}})$.
\end{example}

\begin{figure}[tbp]
\centering\includegraphics[scale=0.55]{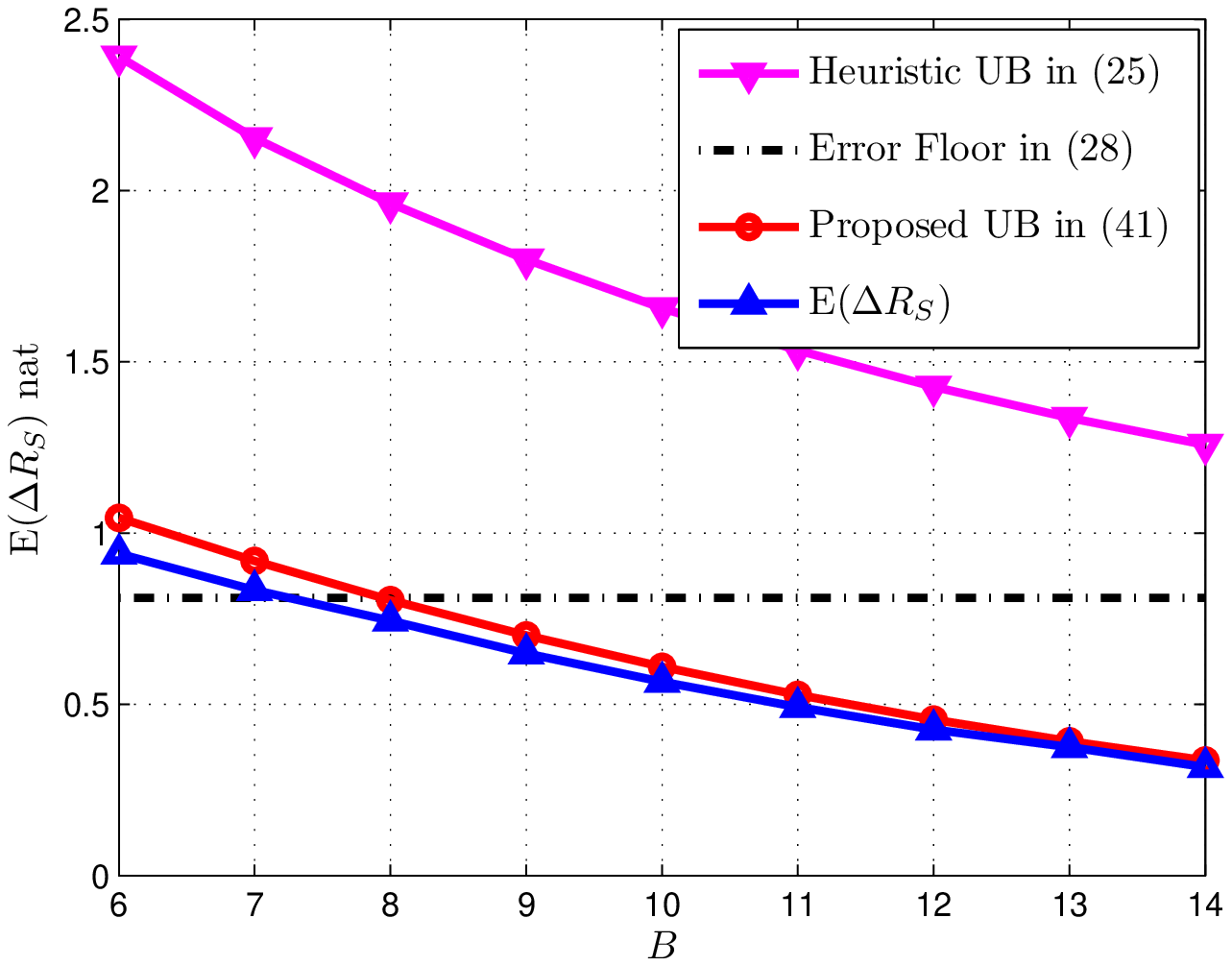} \vspace{-3 mm}
\caption{E$(\Delta R_{\text{S}})$ vs. $B$ with $\protect\beta =1$, $\protect%
\alpha \protect\gamma =1$, $N_{\text{A}}=4$ and $N_{\text{B}}=2$.}
\end{figure}

We then study the distribution of instantaneous secrecy rate loss, defined by%
\begin{equation}
\Delta R_{\text{S}}\triangleq R_{\text{S}}-R_{\text{S,Q}}\text{.}
\end{equation}%
Here, we consider the large system limit as $N_{\text{A}}$ and $B\rightarrow
\infty $ with fixed ratio $B/N_{\text{A}}$. An interesting case that leads
to a closed-form bound can be found when $N_{\text{B}}=N_{\text{E}}=1$.

\begin{theorem}
\label{Th3}If $N_{\text{B}}=N_{\text{E}}=1$, as $N_{\text{A}}$ and $%
B\rightarrow \infty $ with $B/N_{\text{A}}\rightarrow \bar{B}$,%
\begin{eqnarray}
&&\Delta R_{\text{S}}\overset{a.s.}{\rightarrow }\log \left( 1+P/\beta
\right) +\log \left( 1+2^{-\bar{B}}P\right)  \notag \\
&&-\log \left( 1+P+\dfrac{1-\beta }{\beta }(1-2^{-\bar{B}})P\right) \text{.}
\label{app}
\end{eqnarray}
\end{theorem}

\begin{proof}
See Appendix C.
\end{proof}

Theorem \ref{Th3} provides a closed-form asymptotic expression for $\Delta
R_{\text{S}}$ when $N_{\text{B}}=N_{\text{E}}=1$. Hence, the ergodic secrecy
rate loss also converges to the same constant, as stated in the following
corollary.

\begin{corollary}
Under the same assumptions of Theorem \ref{Th3},%
\begin{eqnarray}
&&\mathrm{E}(\Delta R_{\text{S}})\rightarrow \log \left( 1+P/\beta \right)
+\log \left( 1+2^{-\bar{B}}P\right)  \notag \\
&&-\log \left( 1+P+\dfrac{1-\beta }{\beta }(1-2^{-\bar{B}})P\right) \text{.}
\label{app1}
\end{eqnarray}
\end{corollary}

\begin{proof}
The proof is straightforward.
\end{proof}

\begin{example}
The numerical result in Fig.~2 shows that (\ref{app1}) is very
accurate even for finite $N_{\text{A}}$ and $B$.
\begin{figure}[tbp]
\centering\includegraphics[scale=0.55]{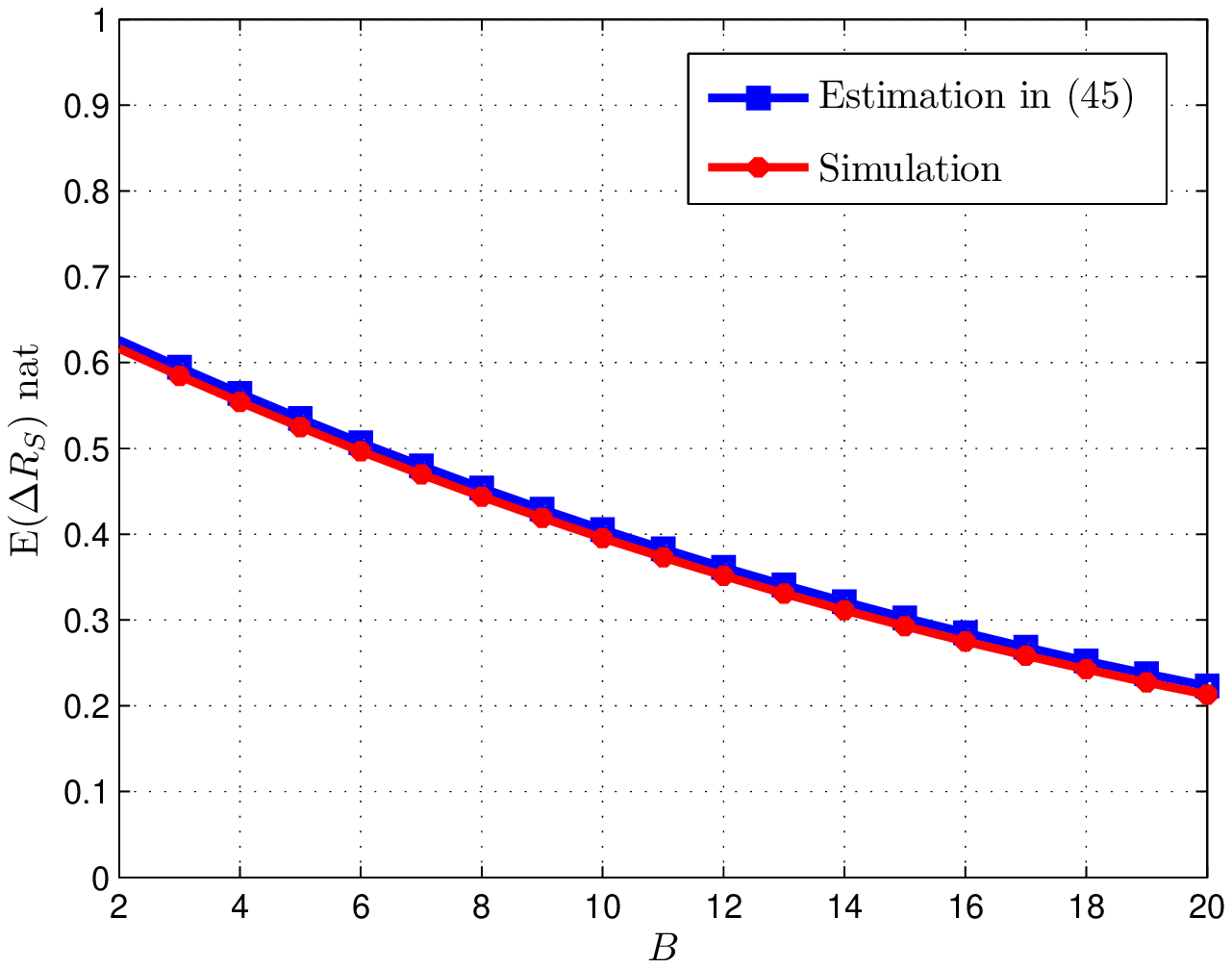}\vspace{-3
mm}
\caption{E$(\Delta R_{\text{S}})$ vs. $B$ with $\protect\beta =1$, $P=1$, $%
N_{\text{A}}=10$ and $N_{\text{B}}=1$.}
\end{figure}
\end{example}

\subsection{A Lower Bound on Ergodic Secrecy Capacity}

A lower bound on $\E(C_{\text{S,Q}})$ can be derived using the results from (%
\ref{Bound_M}), Theorem \ref{Th2} and \cite[Th. 1]{shuiyin14_2}.

\begin{theorem}
\label{Th4}%
\begin{eqnarray}
&&\E(C_{\text{S,Q}})\geq \Theta (N_{\min },N_{\max },\alpha \beta )-\Omega
+\Theta (N_{\text{B}},N_{\text{A}},\theta _{\min })  \notag \\
&&-\Theta \left( N_{\text{B}},N_{\text{A}},\alpha \beta \gamma \dfrac{\eta
\left( N_{\text{A}},N_{\text{B}},2^{B}\right) }{N_{\text{B}}}\right)
\triangleq \bar{C}_{\text{LB,Q}}\text{,}  \label{C_LB_F}
\end{eqnarray}%
where $\Theta (\cdot ,\cdot ,\cdot )$ is given in (\ref{Cap_CF}), $\eta
(\cdot ,\cdot ,\cdot )$ is given in (\ref{niu}) and%
\begin{equation}
\Omega =\left\{
\begin{array}{l}
K\sum\limits_{k=1}^{\hat{N}_{\min }}\det \left( \mathbf{R}^{(k)}\right)
\text{, }~~~~~~\beta \neq 1 \\
\Theta (\hat{N}_{\min },\hat{N}_{\max },\alpha )\text{, }~~~~~~~\beta =1%
\end{array}%
\right.  \label{PPP}
\end{equation}%
\begin{equation}
K=\dfrac{(-1)^{N_{\text{E}}(N_{\text{A}}-\hat{N}_{\min })}}{\Gamma _{\hat{N}%
_{\min }}(N_{\text{E}})}\dfrac{\prod\limits_{i=1}^{2}\mu _{i}^{m_{i}N_{\text{%
E}}}}{\prod\limits_{i=1}^{2}\Gamma _{m_{i}}(m_{i})\prod\limits_{i<j}\left(
\mu _{i}-\mu _{j}\right) ^{m_{i}m_{j}}}\text{,}  \label{K}
\end{equation}%
\begin{equation*}
\Gamma _{k}(n)=\prod\limits_{i=1}^{k}(n-i)!\text{,}
\end{equation*}%
and $\mu _{1}>\mu _{2}$ are the two distinct eigenvalues of the matrix $%
\mathrm{diag}\left( \left\{ \theta _{i}^{-1}\right\} _{1}^{N_{\text{A}%
}}\right) $, with corresponding multiplicities $m_{1}$ and $m_{2}$ such that
$m_{1}+m_{2}=N_{\text{A}}$. The matrix $\mathbf{R}^{(k)}$ has elements%
\begingroup\renewcommand*{\arraystretch}{2.5}
\begin{equation}
r_{i,j}^{(k)}=\left\{ \!\!\!\!%
\begin{array}{l}
\displaystyle(\mu _{\displaystyle e_{i}})^{N_{\text{A}}-j-d_{i}}\dfrac{%
\displaystyle\left( \displaystyle N_{\text{A}}-j\right) !}{(N_{\text{A}%
}-j-d_{i})!}\text{, }~~~~~\hat{N}_{\min }+1\leq j\leq N_{\text{A}} \\
\displaystyle(-1)^{\displaystyle d_{i}}\dfrac{\displaystyle\varphi (i,j)!}{%
\displaystyle(\mu _{\displaystyle e_{i}})^{\displaystyle\varphi (i,j)+1}}%
\text{,}~~~~~~~~~~~~1\leq j\leq \hat{N}_{\min }\text{, }j\neq k \\
\displaystyle(-1)^{\displaystyle d_{i}}\varphi (i,j)!e^{\displaystyle\mu _{%
\displaystyle e_{i}}}\sum\limits_{l=0}^{\displaystyle\varphi (i,j)}\dfrac{%
\displaystyle\Gamma (l-\varphi (i,j),\mu _{\displaystyle e_{i}})}{%
\displaystyle(\mu _{\displaystyle e_{i}})^{\displaystyle l+1}}\text{, }\text{%
otherwise}%
\end{array}%
\right.  \label{R}
\end{equation}%
\endgroup where%
\begin{equation*}
e_{i}=\QATOPD\{ . {1~~~~~~~~~~1\leq i\leq m_{1}}{2~~~m_{1}+1\leq i\leq N_{%
\text{A}}}
\end{equation*}%
\begin{equation*}
d_{i}=\sum_{k=1}^{e_{i}}m_{k}-i\text{,}
\end{equation*}%
\begin{equation*}
\varphi (i,j)=N_{\text{E}}-\hat{N}_{\min }+j-1+d_{i}\text{.}
\end{equation*}
\end{theorem}

\begin{proof}
See Appendix D.
\end{proof}

Theorem \ref{Th4} gives a lower bound on $\E(C_{\text{S,Q}})$, for any
number of Alice/Bob/Eve antennas, as well as for any Bob/Eve SNR regimes.
The lower bound in (\ref{C_LB_F}) is an increasing function of the number of
feedback bits $B$. To guarantee a positive secrecy capacity, Alice just
needs to increase $B$ and checks whether $\bar{C}_{\text{LB,Q}}>0$.

\subsection{Positive Secrecy Capacity with Quantized CSI}

To characterize the achievability of positive secrecy capacity with
quantized CSI, we start by analyzing the tightness of (\ref{C_LB_F}).

\begin{theorem}
\label{Th5}If $N_{\text{E}}\leq N_{\text{A}}-N_{\text{B}}$ and $\beta \geq 1$%
, as $\alpha \beta $, $B\rightarrow \infty $,%
\begin{equation}
\bar{C}_{\text{LB,Q}}=\E(C_{\text{S,Q}})=\E(C_{\text{S}})=\bar{C}_{\text{Bob}%
}\text{,}  \label{t3}
\end{equation}%
where $\bar{C}_{\text{Bob}}$ represents Bob's ergodic channel capacity.
\end{theorem}

\begin{proof}
See Appendix E.
\end{proof}

We have shown that $\bar{C}_{\text{LB,Q}}$, $\E(C_{\text{S}})$ and $\bar{C}_{%
\text{Bob}}$ coincide asymptotically as $B$ and $P_{\text{v}}=\alpha \beta
\gamma (N_{\text{A}}-N_{\text{B}})$ go to infinity. We remark that according
to (\ref{CS_ave}), a universal upper bound on the ergodic secrecy capacity
is given by%
\begin{equation}
\E(C_{\text{S}})\leq \max_{p\left( \mathbf{u}\right) }\left\{ I(\mathbf{%
u;z|H)}\right\} =\bar{C}_{\text{Bob}}\text{.}  \label{t31}
\end{equation}

\begin{remark}
Note that $\bar{C}_{\text{Bob}}>0$ and $\bar{C}_{\text{LB,Q}}$ is derived
based on Gaussian input alphabets. From Theorem \ref{Th5}, we state that a
positive secrecy capacity for MIMOME channel with quantized CSI is always
achieved by using RVQ-based AN transmission scheme and Gaussian input
alphabets for large $B$ and $P_{\text{v}}$, if $N_{\text{E}}\leq N_{\text{A}%
}-N_{\text{B}}$.
\end{remark}

\begin{example}
Fig.~3 compares $\bar{C}_{\text{LB,Q}}$ and $\bar{C}_{\text{Bob}}$
as a
function of AN power $P_{\text{v}}$, with $N_{\text{A}}=4$, $N_{\text{B}}=N_{%
\text{E}}=2$, and $\alpha =\gamma =1$. Since $P_{\text{u}}=\alpha \gamma N_{%
\text{B}}$ and $P_{\text{v}}=\alpha \beta \gamma (N_{\text{A}}-N_{\text{B}})$%
, we have $P_{\text{u}}=2$ and $P_{\text{v}}=2\beta $. The simulation result
shows that $\bar{C}_{\text{LB,Q}}$ approaches to $\bar{C}_{\text{Bob}}$ as $%
P_{\text{v}}$ increases, for sufficiently large $B$.
\end{example}

\begin{figure}[tbp]
\centering\includegraphics[scale=0.55]{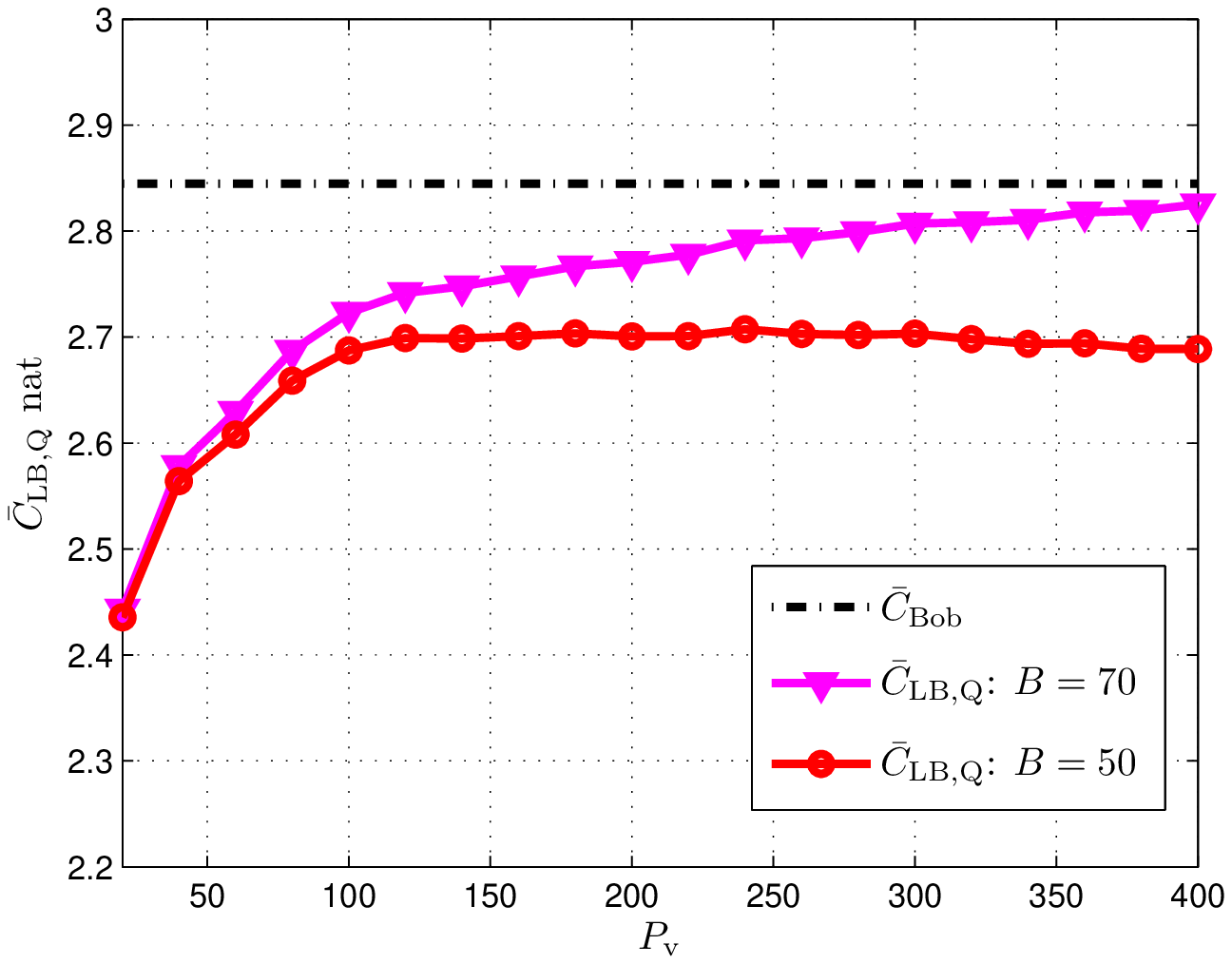} \vspace{-3mm}
\caption{$\bar{C}_{\text{LB,Q}}$ vs. $P_{\text{v}}$ with $N_{\text{A}}=4$, $%
N_{\text{B}}=N_{\text{E}}=2$, and $\protect\alpha =\protect\gamma =1$.}
\end{figure}

\section{Implementation Using a Deterministic Codebook}

In the previous section, random quantization codebooks have been used to
prove new results on secrecy capacity with quantized CSI. The methods of
constructing random unitary matrices $\mathbf{\tilde{V}}_{i}$ in (\ref{V})
can be found in \cite{Karol94}. In practice, it is often desirable that the
quantization codebook is deterministic. The problem of derandomizing RVQ
codebooks is typically referred to as \emph{Grassmannian subspace packing}
\cite{Love03,Mukkavilli03}. Despite of a few special cases (e.g., $B\leq 4$
\cite{Pitaval11}), analytical codebook design in general remains an
intricate task. In this section, we propose a very efficient quantization
codebook construction method for the case of $N_{\text{A}}=2$ and $N_{\text{B%
}}=1$.

According to \cite[Eq. (20)]{Pitaval11}, the codeword $\mathbf{\tilde{V}}%
_{i} $ can be expressed as%
\begin{equation}
\mathbf{\tilde{V}}_{i}(\omega ,\phi )=\left[
\begin{array}{c}
\cos \omega \\
e^{j\phi }\sin \omega%
\end{array}%
\right] \text{,}  \label{Ex}
\end{equation}%
which fully describes the complex Grassmannian manifold $G_{2,1}$ by setting
$0\leq \phi \leq 2\pi $ and $0\leq \omega \leq \pi /2$. Let $(\hat{\omega},%
\hat{\phi})$ be spherical coordinates parameterizing the unit sphere $S^{2}$%
, where $0\leq \hat{\phi}\leq 2\pi $ and $0\leq \hat{\omega}\leq \pi $. In
\cite[Lemma 1]{Pitaval11}, the authors further show that the map%
\begin{eqnarray}
S^{2} &\rightarrow &G_{2,1}  \notag \\
(\hat{\omega},\hat{\phi}) &\longmapsto &\mathbf{\tilde{V}}_{i}(\hat{\omega}%
/2,\hat{\phi})
\end{eqnarray}%
is an isomorphism. In other words, the sampling problem on $G_{2,1}$ can
analogically be addressed on the real sphere $S^{2}$.

The method of sampling points uniformly from $S^{2}$ is provided in \cite%
{Eric}. In details, one can parameterize $(x,y,z)\in S^{2}$ using spherical
coordinates $(\hat{\omega},\hat{\phi})$:%
\begin{eqnarray}
x &=&\sin \hat{\omega}\cos \hat{\phi}\text{,}  \notag \\
y &=&\sin \hat{\omega}\sin \hat{\phi}\text{,}  \notag \\
z &=&\cos \hat{\omega}\text{.}
\end{eqnarray}%
The area element of $S^{2}$ is given by%
\begin{equation}
\D S=\sin \hat{\omega}\D\hat{\omega}\D\hat{\phi}=-\D\left( \cos \hat{\omega}%
\right) \D\hat{\phi}\text{.}
\end{equation}%
Hence, to obtain a uniform distribution over $S^{2}$, one has to pick $\hat{%
\phi}\in \left[ 0,2\pi \right] $ and $t\in \left[ -1,1\right] $ uniformly
and compute $\hat{\omega}$ by:%
\begin{equation}
\hat{\omega}=\arccos t\text{.}
\end{equation}%
In this way $\cos \hat{\omega}=t$ will be uniformly distributed in $\left[
-1,1\right] $.

Based on above analysis, we give a straightforward method for codebook
construction:%
\begin{equation}
\mathcal{\hat{V}}=\Bigg\{\left. \mathbf{\hat{V}}_{1,i}=\left[
\begin{array}{c}
\cos (0.5\arccos t_{i}) \\
e^{j\phi _{i}}\sin (0.5\arccos t_{i})%
\end{array}%
\right] \right\vert i=1,\text{...},2^{B}\Bigg\}\text{,}  \label{D_Codebook}
\end{equation}%
where%
\begin{eqnarray}
t_{i} &=&-1+\dfrac{2\left\lceil i/2^{\left\lceil B/2\right\rceil
}\right\rceil -1}{2^{\left\lfloor B/2\right\rfloor }}\text{,} \\
\phi _{i} &=&\dfrac{2\pi \left( i\func{mod}2^{\left\lceil B/2\right\rceil
}\right) }{2^{\left\lceil B/2\right\rceil }}\text{.}
\end{eqnarray}%
Note that $\lfloor x\rfloor $ rounds to the closest integer smaller than or
equal to $x$, while $\left\lceil x\right\rceil $ to the closest integer
larger than or equal to $x$.

Using the deterministic codebook in (\ref{D_Codebook}) can save storage
space on Alice, since she can generate the target codeword $\mathbf{\hat{V}}%
_{1,j}$ directly without the knowledge of the whole codebook $\mathcal{\hat{V%
}}$. We remark that the proposed codebook construction is valid for any $B$.
This is different from the construction scheme in \cite[Sec. VI-A]{Pitaval11}%
, which is only possible for the case of $B\leq 4$.

\begin{example}
Fig.~4 examines the performance of the proposed codebook construction with $%
\beta =2$, $\gamma =1$, $P=10$, and $N_{\text{E}}=1$. When $B \leq 4 $, it is seen that
the performance of codebook $\mathcal{\hat{V}}$ in
(\ref{D_Codebook}) is indistinguishable from the optimal one in
\cite[Sec. VI-A]{Pitaval11}. When $B\geq8$, the proposed codebook
provides the same performance as the random one in (\ref{V}).
\end{example}

\begin{figure}[tbp]
\centering\includegraphics[scale=0.55]{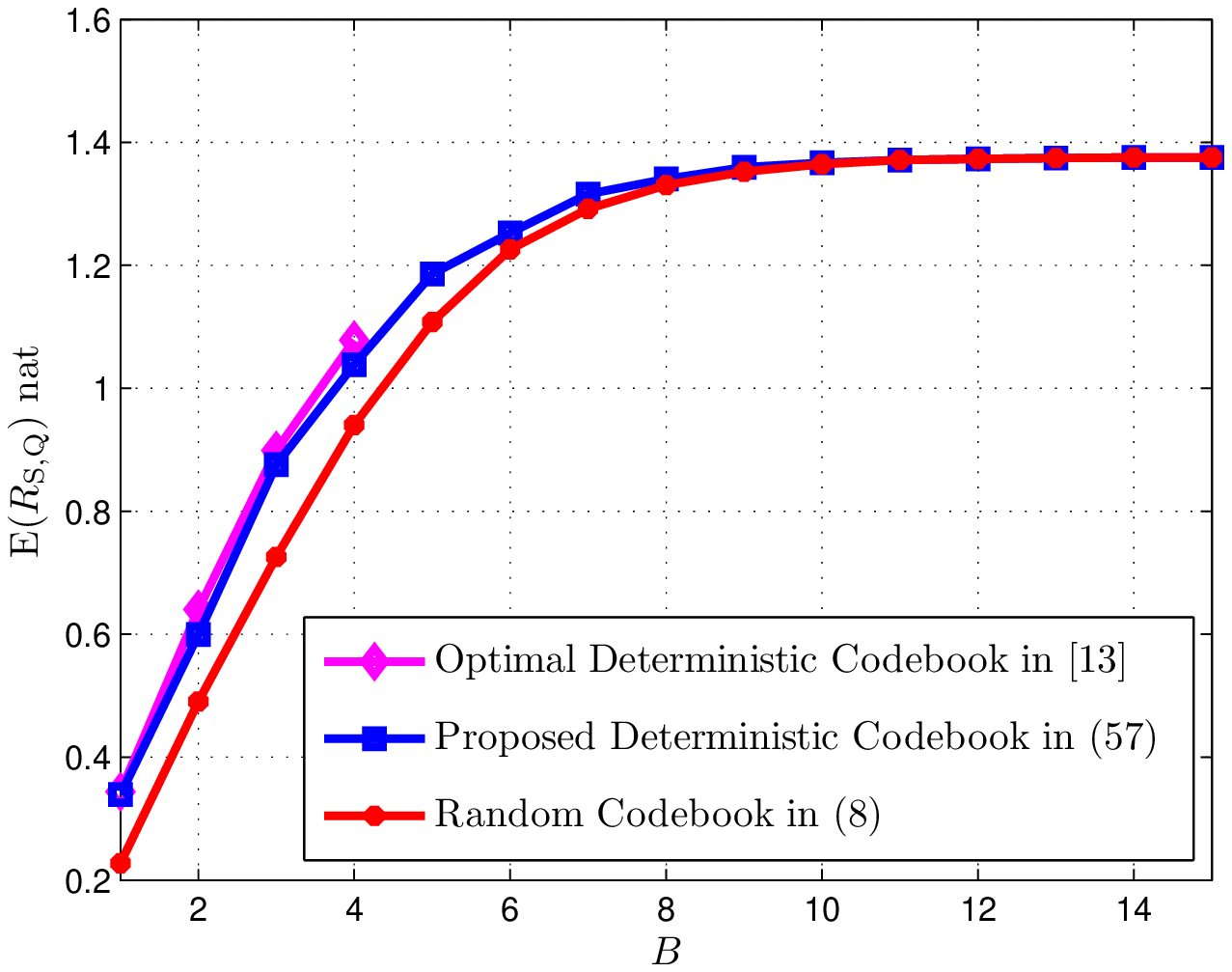}\vspace{-3mm}
\caption{E$(R_{\text{S,Q}})$ vs. $B$ with $\protect\beta =2$, $\protect%
\gamma =1$, $P=10$, $N_{\text{A}}=2$, and $N_{\text{B}}=N_{\text{E}}=1$.}
\end{figure}

\section{Conclusions}

In this work, we have discussed the problem of guaranteeing positive secrecy
capacity for MIMOME channel with the quantized CSI of Bob's channel and the
statistics of Eve's channel. We analyzed the RVQ-based AN scheme and
provided a lower bound on the ergodic secrecy capacity. We proved that a
positive secrecy capacity is always achievable by Gaussian input alphabets
when $N_{\text{E}}\leq N_{\text{A}}-N_{\text{B}}$, and the number of
feedback bits $B$ and the artificial noise power $P_{\text{v}}$ are large
enough. We also proposed an efficient implementation of discretizing the RVQ
codebook which exhibits similar performance to that of random codebook.

\section*{Appendix}

\subsection{Proof of Theorem 1}

According to \cite{Santipach09}, as $B\rightarrow \infty $, the RVQ
operation in (\ref{cost_fun}) can guarantee%
\begin{equation}
\mathbf{\tilde{V}}_{j}\rightarrow \mathbf{\tilde{V}}\text{.}  \label{1}
\end{equation}

We then check the matrix $\mathbf{\hat{Z}}$ generated by Alice. The SVD
decomposition of $\mathbf{H}$ can be written as
\begin{equation}
\mathbf{H=U}\mathbf{\Lambda }\mathbf{\tilde{V}}^{H}\text{.}  \label{2}
\end{equation}

From (\ref{1}) and (\ref{2}), as $B\rightarrow \infty $, we have%
\begin{equation}
\mathbf{H\hat{Z}}=\mathbf{U}\mathbf{\Lambda }\mathbf{\tilde{V}}^{H}\mathbf{%
\hat{Z}\rightarrow U}\mathbf{\Lambda }\mathbf{\tilde{V}}_{j}^{H}\mathbf{\hat{%
Z}=0}_{N_{\text{A}}\times (N_{\text{A}}-N_{\text{B}})}\text{,}
\end{equation}%
which means%
\begin{equation}
\mathbf{\hat{Z}}\rightarrow \mathrm{null}(\mathbf{H})\text{.}  \label{3}
\end{equation}

From (\ref{V_hat}), (\ref{1}) and (\ref{3}), we have $\mathbf{\hat{V}%
\rightarrow V}$ as $B\rightarrow \infty $. \QEDA

\subsection{Proof of Theorem 2}

Using \cite[Eq. 12, pp. 55]{Helmut96}, we have%
\begin{eqnarray}
&&\left\vert \mathbf{I}_{N_{\text{B}}}+\alpha \gamma (\mathbf{H\tilde{V}}%
_{j})(\mathbf{H\tilde{V}}_{j})^{H}+\alpha \beta \gamma (\mathbf{H\hat{Z})}(%
\mathbf{H\hat{Z})}^{H}\right\vert  \notag \\
&\geq &\left\vert \mathbf{I}_{N_{\text{B}}}+\theta _{\min }(\mathbf{H\tilde{V%
}}_{j})(\mathbf{H\tilde{V}}_{j})^{H}+\theta _{\min }(\mathbf{H\hat{Z})}(%
\mathbf{H\hat{Z})}^{H}\right\vert  \notag \\
&=&\left\vert \mathbf{I}_{N_{\text{B}}}+\theta _{\min }\mathbf{HH}%
^{H}\right\vert \text{,}  \label{Main2}
\end{eqnarray}%
where $\theta _{\min }=\min \left\{ \alpha \gamma ,\alpha \beta \gamma
\right\} $.

Since the unitary matrix $\mathbf{\hat{V}=[\tilde{V}}_{j},\mathbf{\hat{Z}]}$
is independent of $\mathbf{G}$ and its realization is known to Alice, $%
\mathbf{G\tilde{V}}_{j}\in \mathbb{C}^{N_{\text{E}}\times N_{\text{B}}}$ and
$\mathbf{G\hat{Z}}\in \mathbb{C}^{N_{\text{E}}\times \left( N_{\text{A}}-N_{%
\text{B}}\right) }$ are mutually independent complex Gaussian random
matrices with i.i.d. entries \cite[Th. 1]{Lukacs54}. We can write%
\begin{eqnarray}
\mathrm{E}\left( \log \dfrac{\left\vert \mathbf{I}_{N_{\text{E}}}+\alpha (%
\mathbf{G\tilde{V}}_{j})(\mathbf{G\tilde{V}}_{j})^{H}+\alpha \beta (\mathbf{G%
\hat{Z})}(\mathbf{G\hat{Z})}^{H}\right\vert }{\left\vert \mathbf{I}_{N_{%
\text{E}}}+\alpha \beta (\mathbf{G\hat{Z})}(\mathbf{G\hat{Z})}%
^{H}\right\vert }\right)  \label{B01}
\end{eqnarray}%
as the average of a function of $N_{\text{E}} \times N_{\text{A}}$ i.i.d
complex Gaussian random variables $\sim \mathcal{N}_{\mathbb{C}}(0,1)$.

Similarly, with unlimited feedback, we have%
\begin{eqnarray}
\mathrm{E}\left( \log \dfrac{\left\vert \mathbf{I}_{N_{\text{E}}}+\alpha (%
\mathbf{G\tilde{V}})(\mathbf{G\tilde{V}})^{H}+\alpha \beta (\mathbf{GZ)}(%
\mathbf{GZ)}^{H}\right\vert }{\left\vert \mathbf{I}_{N_{\text{E}}}+\alpha
\beta (\mathbf{GZ)}(\mathbf{GZ)}^{H}\right\vert }\right)  \label{B02}
\end{eqnarray}%
as the average of a function of $N_{\text{E}} \times N_{\text{A}}$ i.i.d
complex Gaussian random variables $\sim \mathcal{N}_{\mathbb{C}}(0,1)$.

From (\ref{B01}) and (\ref{B02}), we have%
\begin{eqnarray}
&&\mathrm{E}\left( \log \dfrac{\left\vert \mathbf{I}_{N_{\text{E}}}+\alpha (%
\mathbf{G\tilde{V}}_{j})(\mathbf{G\tilde{V}}_{j})^{H}+\alpha \beta (\mathbf{G%
\hat{Z})}(\mathbf{G\hat{Z})}^{H}\right\vert }{\left\vert \mathbf{I}_{N_{%
\text{E}}}+\alpha \beta (\mathbf{G\hat{Z})}(\mathbf{G\hat{Z})}%
^{H}\right\vert }\right)  \notag \\
&=&\mathrm{E}\left( \log \dfrac{\left\vert \mathbf{I}_{N_{\text{E}}}+\alpha (%
\mathbf{G\tilde{V}})(\mathbf{G\tilde{V}})^{H}+\alpha \beta (\mathbf{GZ)}(%
\mathbf{GZ)}^{H}\right\vert }{\left\vert \mathbf{I}_{N_{\text{E}}}+\alpha
\beta (\mathbf{GZ)}(\mathbf{GZ)}^{H}\right\vert }\right) \text{.}
\label{B03}
\end{eqnarray}

From (\ref{SR_P_CSI}), (\ref{SR_Q_CSI}), (\ref{Main2}) and (\ref{B03}), $%
\mathrm{E}(\Delta R_{\text{S}})$ can be upper bounded by%
\begin{align}
\mathrm{E}(\Delta R_{\text{S}})& \leq \mathrm{E}\left( \log \left\vert
\mathbf{I}_{N_{\text{B}}}+\alpha \gamma \mathbf{HH}^{H}\right\vert \right)
\notag \\
& -\mathrm{E}\left( \log \left\vert \mathbf{I}_{N_{\text{B}}}+\theta _{\min }%
\mathbf{HH}^{H}\right\vert \right)  \notag \\
& +\mathrm{E}\left( \log \left\vert \mathbf{I}_{N_{\text{B}}}+\alpha \beta
\gamma (\mathbf{H\hat{Z})}(\mathbf{H\hat{Z})}^{H}\right\vert \right) \text{.}
\label{M111}
\end{align}

We then estimate the third term in (\ref{M111}). Let $\lambda _{1}$, ... , $%
\lambda _{N_{\text{B}}}$ be the eigenvalues of $\mathbf{HH}^{H}$. We have%
\begin{equation}
\mathbf{H}^{H}\mathbf{H}=\mathbf{\tilde{V}}\mathbf{\Lambda \tilde{V}}^{H}%
\text{ and }\mathbf{\Lambda }=\mathrm{diag}\left( \left[ \lambda
_{1},...,\lambda _{N_{\text{B}}}\right] \right) \text{.}
\end{equation}%
Recalling the fact that for a Wishart matrix, its eigenvalues and
eigenvectors are independent. Therefore $\mathbf{\tilde{V}}$ and $\mathbf{%
\Lambda }$ are independent. This allows us to bound the third term in (\ref%
{M111}) by%
\begin{align}
& \mathrm{E}_{\mathbf{H}}\left( \log \left\vert \mathbf{I}_{N_{\text{B}%
}}+\alpha \beta \gamma (\mathbf{H\hat{Z})}(\mathbf{H\hat{Z})}^{H}\right\vert
\right)  \notag \\
& \overset{(a)}{=}\mathrm{E}_{\mathbf{H}}\left( \log \left\vert \mathbf{I}%
_{N_{\text{A}}-N_{\text{B}}}+\alpha \beta \gamma (\mathbf{H\hat{Z})}^{H}(%
\mathbf{H\hat{Z})}\right\vert \right)  \notag \\
& =\mathrm{E}_{\Lambda }\left( \mathrm{E}_{\mathbf{\tilde{V}}}\left( \log
\left\vert \mathbf{I}_{N_{\text{A}}-N_{\text{B}}}+\alpha \beta \gamma
\mathbf{\hat{Z}}^{H}\mathbf{\tilde{V}}\mathbf{\Lambda \tilde{V}}^{H}\mathbf{%
\hat{Z}}\right\vert \right) \right)  \notag \\
& \overset{(b)}{=}\mathrm{E}_{\Lambda }\left( \mathrm{E}_{\mathbf{\tilde{V}}%
}\left( \log \left\vert \mathbf{I}_{N_{\text{B}}}+\alpha \beta \gamma
\mathbf{\tilde{V}}^{H}\mathbf{\hat{Z}\hat{Z}}^{H}\mathbf{\tilde{V}}\mathbf{%
\Lambda }\right\vert \right) \right)  \notag \\
& \overset{(c)}{\leq }\mathrm{E}_{\Lambda }\left( \mathrm{E}_{\mathbf{\tilde{%
V}}}\left( \log \left\vert \mathbf{I}_{N_{\text{B}}}+\alpha \beta \gamma
\mathrm{E}_{\mathbf{\tilde{V}}}\left( \mathbf{\tilde{V}}^{H}\mathbf{\hat{Z}%
\hat{Z}}^{H}\mathbf{\tilde{V}}\right) \mathbf{\Lambda }\right\vert \right)
\right)  \notag \\
& \overset{(d)}{=}\mathrm{E}_{\Lambda }\left( \log \left\vert \mathbf{I}_{N_{%
\text{B}}}+\dfrac{\alpha \beta \gamma D\left( N_{\text{A}},N_{\text{B}%
},2^{B}\right) }{N_{\text{B}}}\mathbf{\Lambda }\right\vert \right)  \notag \\
& \overset{(e)}{\leq }\mathrm{E}_{\Lambda }\left( \log \left\vert \mathbf{I}%
_{N_{\text{B}}}+\dfrac{\alpha \beta \gamma \eta \left( N_{\text{A}},N_{\text{%
B}},2^{B}\right) }{N_{\text{B}}}\mathbf{\Lambda }\right\vert \right) \text{,}
\label{m23}
\end{align}%
where $(a)$ and $(b)$ hold because $|\mathbf{I+AB}|=|\mathbf{I+BA}|$, $(c)$
follows from the concavity of log determinant function, $(d)$ follows from
\cite[Lemma 1]{Jindal08}\cite[Lemma 2]{Shihc11}, $(e)$ holds because of (\ref%
{bbound}).

Applying the fact \cite[Th. 1]{shin03}%
\begin{equation}
\mathrm{E}\left( \log \left\vert \mathbf{I}_{N_{\text{B}}}+\rho \mathbf{HH}%
^{H}\right\vert \right) =\mathrm{E}\left( \log \left\vert \mathbf{I}_{N_{%
\text{B}}}+\rho \mathbf{\Lambda }\right\vert \right) =\Theta (N_{\text{B}%
},N_{\text{A}},\rho )\text{,}
\end{equation}%
to (\ref{M111}) and (\ref{m23}), we can simply obtain (\ref{MMO}). \QEDA

\subsection{Proof of Theorem 3}

Recalling the fact that $\mathbf{ZZ}^{H}=\mathbf{I}_{N_{\text{A}}}-\mathbf{%
\tilde{V}\tilde{V}}^{H}$ and $\mathbf{\hat{Z}\hat{Z}}^{H}=\mathbf{I}_{N_{%
\text{A}}}-\mathbf{\tilde{V}}_{j}\mathbf{\tilde{V}}_{j}^{H}$\textbf{. }From (%
\ref{SR_P_CSI}) and (\ref{SR_Q_CSI}), if $N_{\text{B}}=N_{\text{E}}=1$, we
can write $\Delta R_{\text{S}}$ as%
\begin{align}
\Delta R_{\text{S}}& =\log \left( 1+\alpha \gamma \mathbf{HH}^{H}\right)
\notag \\
& -\log \dfrac{1+\alpha \mathbf{GG}^{H}+\alpha (\beta -1)(\mathbf{GZ)}(%
\mathbf{GZ)}^{H}}{1+\alpha \beta (\mathbf{GZ)}(\mathbf{GZ)}^{H}}  \notag \\
& -\log \dfrac{1+\alpha \beta \gamma \mathbf{HH}^{H}+\alpha \gamma \left(
1-\beta \right) (\mathbf{H\tilde{V}}_{j})(\mathbf{H\tilde{V}}_{j})^{H}}{%
1+\alpha \beta \gamma \mathbf{HH}^{H}-\alpha \beta \gamma (\mathbf{H\tilde{V}%
}_{j})(\mathbf{H\tilde{V}}_{j})^{H}}  \notag \\
& +\log \dfrac{1+\alpha \mathbf{GG}^{H}+\alpha (\beta -1)(\mathbf{G\hat{Z})}(%
\mathbf{G\hat{Z})}^{H}}{1+\alpha \beta (\mathbf{G\hat{Z})}(\mathbf{G\hat{Z})}%
^{H}}\text{.}  \label{d1}
\end{align}

As $N_{\text{A}}$ and $B\rightarrow \infty $ with $B/N_{\text{A}}\rightarrow
\bar{B}$, according to \cite[Th. 1]{Santipach09}, we have%
\begin{equation}
\dfrac{(\mathbf{H\tilde{V}}_{j})(\mathbf{H\tilde{V}}_{j})^{H}}{\mathbf{HH}%
^{H}}\overset{a.s.}{\rightarrow }(1-2^{-\bar{B}}).  \label{5}
\end{equation}

Since $\beta $ (AN power allocation), $\gamma $ (Eve-to-Bob noise-power
ratio), and $P=\alpha \gamma +\alpha \beta \gamma (N_{\text{A}}-1)$ (average
transmit power constraint) are fixed, the central limit theorem tells us that%
\begin{align}
& \alpha \gamma \mathbf{HH}^{H}\overset{a.s.}{\rightarrow }P/\beta \text{, }%
\alpha \mathbf{GG}^{H}\overset{a.s.}{\rightarrow }P/\beta \gamma \text{,}
\notag \\
& \alpha (\mathbf{GZ)}(\mathbf{GZ)}^{H}\overset{a.s.}{\rightarrow }P/\beta
\gamma \text{, }\alpha (\mathbf{G\hat{Z})}(\mathbf{G\hat{Z})}^{H}\overset{%
a.s.}{\rightarrow }P/\beta \gamma \text{.}  \label{4}
\end{align}%
Note that $\mathbf{GZ}$ (or $\mathbf{G\hat{Z}}$) is a complex Gaussian
random vector with i.i.d. entries \cite[Th. 1]{Lukacs54}.

By substituting (\ref{5}) and (\ref{4}) into (\ref{d1}), we obtain (\ref{app}%
). \QEDA

\subsection{Proof of Theorem 4}

According to \cite[Th. 1]{shuiyin14_2}, we have%
\begin{equation}
\E(R_{\text{S}})=\Theta (N_{\text{B}},N_{\text{A}},\alpha \gamma )+\Theta
(N_{\min },N_{\max },\alpha \beta )-\Omega \text{,}  \label{L1}
\end{equation}%
where $\Theta (\cdot ,\cdot ,\cdot )$ is given in (\ref{Cap_CF}) and $\Omega
$ is given in (\ref{PPP}). By substituting (\ref{MMO}) and (\ref{L1}) into (%
\ref{Bound_M}), we can obtain (\ref{C_LB_F}). \QEDA

\subsection{Proof of Theorem 5}

If $\beta \geq 1$, then $\theta _{\min }=\alpha \gamma $. From (\ref{C_LB_F}%
) and (\ref{L1}), as $B\rightarrow \infty $,%
\begin{equation}
\bar{C}_{\text{LB,Q}}=\Theta (N_{\text{B}},N_{\text{A}},\alpha \gamma
)+\Theta (N_{\min },N_{\max },\alpha \beta )-\Omega =\E(R_{\text{S}})\text{.}
\label{5_1}
\end{equation}

According to \cite[Th. 3]{shuiyin14_2}, if $N_{\text{E}}\leq N_{\text{A}}-N_{%
\text{B}}$, as $\alpha \beta \rightarrow \infty $,%
\begin{equation}
\E(R_{\text{S}})=\E(C_{\text{S}})=\bar{C}_{\text{Bob}}\text{,}  \label{5_2}
\end{equation}%
where $\bar{C}_{\text{Bob}}$ represents Bob's average channel capacity.

Meanwhile, it always holds that%
\begin{equation}
\bar{C}_{\text{LB,Q}}\leq \E(C_{\text{S,Q}})\leq \bar{C}_{\text{Bob}}\text{.}
\label{5_3}
\end{equation}

From (\ref{5_1}), (\ref{5_2}) and (\ref{5_3}), if $N_{\text{E}}\leq N_{\text{%
A}}-N_{\text{B}}$ and $\beta \geq 1$, as $\alpha \beta $, $B\rightarrow
\infty $, we have%
\begin{equation}
\bar{C}_{\text{LB,Q}}=\E(C_{\text{S,Q}})=\E(C_{\text{S}})=\bar{C}_{\text{Bob}%
}\text{.}
\end{equation}%
\QEDA

\bibliographystyle{IEEEtran}
\bibliography{IEEEabrv,LIUBIB}

\end{document}